\documentclass[conference]{IEEEtran}

\usepackage{newtxtext,newtxmath}
\usepackage{amsfonts}

\usepackage{amsmath,amssymb,amsthm}
\usepackage{graphicx}
\usepackage{hyperref}
\usepackage{algorithm}
\usepackage{algpseudocode}
\usepackage{booktabs}
\usepackage{microtype}
\usepackage{tikz}
\usepackage{pgfplots}
\usepackage{pifont}
\usetikzlibrary{arrows.meta,positioning,fit}
\usepackage{comment}
\pgfplotsset{compat=1.18}

\newtheorem{theorem}{Theorem}

\begin{document}

\title{
TrustRAG: Blockchain-Enhanced RAG via Committee-Based Credibility Scoring
}

% TODO: fill in author names, affiliations, and emails before uploading to arXiv.
% Example for multiple authors sharing an affiliation:
%
% \author{\IEEEauthorblockN{First Author\IEEEauthorrefmark{1}, Second Author\IEEEauthorrefmark{2}}
% \IEEEauthorblockA{\IEEEauthorrefmark{1}School of Computing and Intelligent Innovation, Fudan University, Shanghai, China\\
% Email: name@fudan.edu.cn}
% \IEEEauthorblockA{\IEEEauthorrefmark{2}Affiliation Two\\
% Email: name@example.edu}}
%\author{\IEEEauthorblockN{Author Name(s)}
%\IEEEauthorblockA{Affiliation\\
%Email: author@example.edu}}
\author{\IEEEauthorblockN{Baixiang Liu}
	\IEEEauthorblockA{Fudan University\\
		bxliu@fudan.edu.cn}
	\and
	\IEEEauthorblockN{Haotian Che}
	\IEEEauthorblockA{Fudan University\\
		23210240106@m.fudan.edu.cn}
	\and
	\IEEEauthorblockN{Yuan Li}
	\IEEEauthorblockA{Fudan University\\
		yuan\_li@fudan.edu.cn}
}

\maketitle

\begin{abstract}

Retrieval-Augmented Generation (RAG) lets Large Language Models (LLMs) pull in up-to-date, domain-specific information instead of relying only on what they were trained on. Yet most RAG systems still draw from centralized databases with limited oversight, making it difficult to verify where a document came from, whether it has been tampered with, or whether it should be trusted at all. 
This is a serious problem in domains where both the timeliness and accuracy of retrieved content are critical, such as healthcare, finance, logistics, and legal case law, where a wrong or manipulated document can directly lead to bad decisions.

We present TrustRAG, a committee-based, blockchain-backed RAG system: before a document is used, it is certified by a committee of domain experts through a zero-knowledge protocol, and the committee's hidden scores are combined via secure multi-party computation into a trust score that any client can verify. These scores, along with the underlying document data, are maintained jointly across chains through hash commitments, so no document or score can be silently altered or dropped, and every ranking can be independently replayed and checked.

\end{abstract}

\section{Introduction}

Large Language Models have achieved remarkable capabilities across reasoning, generation, and knowledge-intensive tasks, yet their reliance on static training datasets fundamentally constrains their ability to provide accurate, up-to-date, and verifiable information. Addressing this limitation, Retrieval-Augmented Generation (RAG) has emerged as a promising paradigm that enriches LLMs with external knowledge retrieval, thereby enhancing factual accuracy and contextual relevance \cite{lewis2020retrieval}.
RAG promises more accurate, up-to-date, and domain-specific responses. This paradigm has therefore been widely adopted in areas such as healthcare, finance, and law, where strong generative capabilities must coexist with strict requirements on privacy, compliance, and auditability.

However, despite its growing role in improving factual grounding, the architecture of modern RAG systems exhibits a structural vulnerability: the retrieval pipelines, indexing services, and augmentation logic remain heavily centralized. Data flows, update policies, and ranking heuristics are typically controlled by a small number of providers or internal platform teams, making it difficult for downstream users or regulators to inspect, verify, or challenge the retrieval layer's behavior. As a result, the very mechanism meant to reduce hallucinations and ground outputs in fact is itself implemented as an opaque subsystem.

The evolution of RAG infrastructures echoes earlier trends in Internet architecture: systems that begin with open, protocol-based principles often drift toward platform-centric consolidation. Early protocols such as TCP/IP were designed to distribute control and enable interoperability without centralized gatekeepers, yet were gradually overshadowed by proprietary platforms that centralized data ownership and shaped information flows through opaque algorithms \cite{garg2019copa,abreu2025should}. RAG infrastructures now show the same trajectory: a modular, evidence-driven mechanism increasingly resembles a closed, platform-governed subsystem whose behavior cannot be independently audited or verified.

This opacity gives rise to two challenges in contemporary RAG systems: evidence trustworthiness and process governance. The first is that retrieved sources are opaque, making it hard to verify their integrity, legality, or authority. Centralized pipelines offer no reliable way to confirm where a passage originated, whether it was modified, or whether it complies with copyright and data-protection requirements. This opacity is exploitable: attackers can dilute relevant evidence by injecting superficially similar but irrelevant text~\cite{chen2024densex,fang2024ragbench}, or distribute contradictory statements across sources to obscure authoritative information and induce ambiguous or incorrect outputs~\cite{wu2024clasheval,liu2023recall,zhou2024trustworthiness}. Without verifiable audit trails or cryptographic provenance, it is difficult to tell whether anomalous retrieval behavior stems from system limitations, data-quality issues, or deliberate manipulation.

The second challenge is that the retrieval-and-augmentation layer itself is opaque, obscuring why a system produces harmful outputs or refuses to respond. This layer acts as a gatekeeper over what conditions the model's responses, yet centralized ranking and filtering heuristics offer little visibility into how conflicting or harmful content is prioritized or suppressed. Such opacity introduces further risks: insufficient filtering can let unsafe or toxic content shape generations~\cite{deshpande2023toxicity,perez2022ignore}, while targeted denial-of-service cues can push RAG pipelines into refusing to answer even when relevant evidence is available~\cite{chaudhari2024phantom,shafran2024mar}. Because these behaviors often resemble benign limitations, operators struggle to pinpoint the cause, complicating efforts to improve robustness and accountability.

A further difficulty arises when document trustworthiness must be evaluated collectively, e.g., by distributed validators, expert committees, or multiple stakeholders, rather than by a single curator. Collecting such scores in plaintext risks leaking evaluator preferences or enabling strategic manipulation; aggregating them inside a closed backend merely turns the trust value into another opaque platform output. Trustworthy RAG therefore requires not only decentralized provenance records, but also a privacy-preserving mechanism for combining distributed trust inputs into verifiable retrieval metadata. %We refer to the resulting system as \emph{TrustRAG}.

Taken together, these developments reveal a growing tension at the core of contemporary RAG practice. Mainstream AI infrastructure, from model training to retrieval, remains highly centralized, yet the data that matters most, e.g., time-sensitive, high-accuracy, access-controlled, and frequently updated, often originates from inherently distributed and decentralized sources: hospitals, regulators, financial institutions, and expert communities that do not share a single trusted curator. Reconciling this mismatch calls for a blockchain-enhanced RAG architecture, one that provides decentralized provenance, privacy-preserving trust aggregation, and verifiable retrieval for high-stakes domains where corrupted evidence or unaccountable filtering carry serious operational and regulatory costs.

To address these challenges, we propose TrustRAG, a decentralized, verifiable, and auditable RAG framework that combines immutable document registration with reusable credibility computation. Documents are first registered with content and embedding hashes as immutable provenance anchors. Validators then submit privacy-preserving quality evaluations bound to document identities and set hashes, enforcing membership, uniqueness, and score correctness without disclosing votes. During consolidation, hidden score components are split into Shamir secret shares and processed by committee nodes via an MP-SPDZ-based secure aggregation interface, yielding chain-local tallies and credibility values without revealing any individual vote or blinding factor. Rather than generating a recursive global aggregation proof, the system binds per-chain outputs through hash commitments and exposes sufficient metadata for deterministic replay of the final ranking. The resulting proofs and aggregate summaries are reused during later retrieval and audit, yielding a verifiable response package while keeping the protocol substantially simpler than recursive-proof alternatives. This design targets domains where knowledge is time-sensitive, accuracy-critical, access-controlled, and continuously updated, such as healthcare, finance, breaking news, logistics and traffic, and legal precedents, where data is naturally distributed across independent, mutually distrusting sources.

\subsection{Related Work}
Recent work has turned to blockchain as a mechanism for improving data reliability and trustworthiness in RAG systems.

Andersen, Avalos, Dagher and Long proposed D-RAG~\cite{e_andersen2025d}, a blockchain-based framework that organizes knowledge into domain-specific communities, where data is validated by field experts prior to database inclusion via a Privacy-Preserving Knowledge Incorporation (PPKI) protocol. PPKI employs zero-knowledge proofs and homomorphic encryption to realize a double-blind consensus, hiding both the proposer's identity and members' votes to mitigate validation bias.
A separate Retrieval Blockchain then coordinates LLM-assisted document ranking and response generation across communities.
Unlike prior work that detects malicious content only at retrieval time, D-RAG addresses data integrity at the source, preventing unsafe data from entering the knowledge base in the first place.

Yu and Sato proposed DeRAG~\cite{yu2024derag}, a decentralized multi-source RAG system that adapts the Pyth Network's oracle technology for knowledge-intensive applications via the RAG-Optimized Pyth Consensus (ROPC) protocol. ROPC extends oracle-based validation to multi-dimensional textual data through tensor-based representations and semantic consistency checks, coordinated by a domain-specialized validator network organized as a directed acyclic graph. This design demonstrates that decentralized oracle mechanisms originally developed for financial data feeds can be effectively repurposed to improve data integrity and retrieval reliability in RAG pipelines.

Lu, Tan, Johnson, Jung and Jiang proposed a decentralized RAG system with a \emph{dynamic reliability scoring mechanism}~\cite{lu2025decentralizedretrievalaugmentedgeneration}, in which each data source is assigned cumulative reliability and usefulness scores updated via sentence-level importance estimation and user feedback.
By incorporating these scores into both source sampling and document reranking, their system progressively prioritizes higher-quality sources without requiring centralized data management. Notably, their approach achieves performance comparable to centralized systems operating on fully reliable data, despite operating over independently maintained, potentially noisy sources.

Lin, Cui, Zhou, et al.~proposed VeriRAG~\cite{VeriRAG2026}, a zero-knowledge proof framework that provides efficient integrity guarantees for the retrieval step of RAG systems. VeriRAG allows a service provider to prove that its returned top-$k$ documents are genuinely the output of an Approximate Nearest Neighbor Search (ANNS) executed faithfully over its committed corpus, without revealing the underlying dataset. To make this practical, the authors introduce a protocol that sidesteps the costly verification of the sorting process itself, together with vector-lookup and chunk-merging optimizations that jointly reduce proving overhead. The resulting system scales to a 37GB dataset with a 96-second prover time and a 3-second verifier time, demonstrating that succinct, privacy-preserving proofs of retrieval correctness are practical even at scale.

Shukla and Joshi proposed Proof-Carrying Answers (PCA)~\cite{shukla2025proofcarrying}, a protocol that attaches a signature and a Merkle proof to every retrieved chunk and admits a generated claim only if all its supporting chunks pass verification, abstaining otherwise. This is an elegant shift in RAG's default posture, from "trust, then verify" to "verify, then trust": cryptographic grounding is checked before an answer is released, rather than being left to post-hoc auditing, and the use of Merkle proofs keeps this check lightweight enough for practical deployment. The protocol relies on a single signer and a Merkle-committed corpus snapshot, treating trustworthiness as a matter of provenance -- that a chunk is unmodified and traceable to a registered source.

\subsection{Our Contributions}

Our paper makes the following contributions.
\begin{enumerate}
\item \textbf{Committee-Certified Knowledge Base.} We introduce a pre-certification architecture in which documents are endorsed by an expert validator committee at registration time, producing reusable trust artifacts that downstream retrieval inherits without repeating the scoring path on every query.

\item \textbf{Privacy-Preserving Multi-Party Trust Aggregation.} We combine the zero-knowledge scoring workflow with Pedersen commitments, Shamir secret sharing, and an MP-SPDZ secure-sum interface so that committee nodes can aggregate hidden validator inputs without revealing individual scores or blinding factors.

\item \textbf{Cross-Chain Binding Without Recursive Proofs.} We replace a global aggregation proof with a lightweight binding layer that hashes each per-chain scoring summary and candidate list into a single global digest, preventing weight substitution or chain omission while reducing system complexity.

\item \textbf{Deterministic Replay for Ranking Verification.} We publish sufficient metadata for any client to independently recompute the ranking scores, reproduce the top-$k$ document selection, and verify that the final answer is consistent with the retrieved evidence, without relying on a recursive proof.
\end{enumerate}

Our core contribution is an end-to-end verifiable pipeline that connects committee-based certification at registration time to cryptographically auditable retrieval at query time. Documents are endorsed by an expert validator committee through a privacy-preserving voting protocol at registration, producing reusable trust artifacts that persist with each document identity. At query time, these artifacts directly inform ranking, and the entire process---from individual votes through score aggregation to final document selection---can be independently replayed and verified by any third party without relying on a trusted intermediary or a recursive aggregation proof. To our knowledge, this is the first decentralized RAG framework that integrates privacy-preserving certification, cross-chain binding, and deterministically replayable ranking into a unified auditable system.

TrustRAG is designed for high-stakes domains where the cost of misinformation is severe and the need for accountable knowledge provenance is paramount. Examples include:
\begin{itemize}
    \item \textbf{Healthcare.} Clinical guidelines and medical literature must be endorsed by qualified experts before influencing diagnostic or treatment recommendations.
    \item \textbf{Legal practice.} Case law and regulatory documents require professional validation to ensure jurisdictional relevance and currency.
    \item \textbf{Finance.} Market data and financial disclosures must be current and traceable to a qualified source, since a stale or manipulated figure can directly distort trading and investment decisions.
    \item \textbf{Logistics and traffic.} Routing, capacity, and disruption information changes continuously, and outdated or unverified reports can propagate through downstream planning decisions.
    \item \textbf{News.} Breaking coverage places a premium on both timeliness and accuracy, and expert or outlet-level endorsement helps prevent unverified or fabricated reports from shaping generated summaries.
    \item \textbf{Government and public administration.} Policy documents and official records must carry verifiable provenance to prevent the citation of forged or superseded versions.
\end{itemize}
\begin{comment}
\begin{itemize}
    \item \textbf{Healthcare.} Clinical guidelines and medical literature must be endorsed by qualified experts before influencing diagnostic or treatment recommendations.
    \item \textbf{Legal practice.} Case law and regulatory documents require professional validation to ensure jurisdictional relevance and currency.
    \item \textbf{Academic research.} A committee-based endorsement mechanism mirrors the peer-review process, guarding against retracted or pseudoscientific material entering the retrieval pool.
    \item \textbf{Government and public administration.} Policy documents and official records must carry verifiable provenance to prevent the citation of forged or superseded versions.
    \item \textbf{Education.} Instructional materials and assessment content require expert certification to ensure factual accuracy.
\end{itemize}
\end{comment}

Across all these settings, TrustRAG provides a framework in which domain experts collectively certify knowledge at the source, and that certification remains cryptographically bound to every subsequent retrieval and ranking decision.

\section{Challenges and Design Overview}
\label{sec:challenges}

Today's RAG systems face a fundamental governance problem: \emph{retrieval is centralized, trust signals are opaque, and the ranking process is unauditable.} Addressing this requires resolving three concrete technical challenges.

\paragraph*{Challenge 1: secure candidate scoring under untrusted validators.}
Trust metadata should be derived from distributed validators rather than from a single curator, but the system must prevent unauthorized participation, duplicate scoring, and out-of-range manipulation within each scoring round bound to a designated document set. At the same time, individual scores and validator identities should remain hidden.

\paragraph*{Challenge 2: privacy-preserving credibility with public verifiability.}
Even if individual votes are hidden, the system still needs a publicly auditable way to prove that each finalized credibility value is consistent with the submitted commitments and with the exact document identities used at retrieval time.

\paragraph*{Challenge 3: cross-chain binding and ranking verification without recursive proofs.}
Once per-chain credibility values are available, the system must still prevent an untrusted service from changing candidate sets, omitting required chains from a declared retrieval round, or reranking documents incorrectly. A practical design should therefore provide end-to-end verifiability without incurring the implementation and proving complexity of a recursive global proof.

TrustRAG addresses these challenges by combining immutable document registration, reusable zero-knowledge scoring, committee-side MPC aggregation, and replay-oriented verification. Validators first prove authorized scoring over registered document identifiers using zero knowledge. Their hidden score components are then aggregated through committee-side MPC, and the resulting document-level values are checked against the underlying commitments and exposed through finalized chain records. These scores are finalized ahead of query time and remain bound to stable document identities rather than to any single future query. At retrieval time, the service derives a candidate set, hashes the exact returned list, fetches the already finalized per-document scores for that list, and binds them through the Aggregator.
Finally, the retrieval layer publishes the finalized scores, and the associated zero-knowledge proofs, allowing any client to verify the validity of the scoring and aggregation and to confirm that the returned answer hash is correctly bound to them.

Compared with representative centralized and decentralized RAG architectures discussed in the literature, TrustRAG emphasizes private candidate scoring, chain-local verifiable tallying, cross-chain binding, and proof-based verification.

\section{Preliminaries}
\label{sec:preliminaries}

\subsection{Retrieval-Augmented Generation}
\label{subsec:rag}
Retrieval-Augmented Generation (RAG)~\cite{lewis2020retrieval} enhances Large Language Models (LLMs) by combining external knowledge retrieval with text generation. Given a query $Q$ and a document collection $\mathcal{D}=\{d_1,\dots,d_n\}$, the retriever applies an encoder $f(\cdot)$ to obtain dense vector representations $\mathbf{q}=f(Q)$ and $\mathbf{d}_i=f(d_i)$, and computes their similarity via cosine similarity:
\[
\text{sim}(\mathbf{q},\mathbf{d}_i)
=\frac{\mathbf{q}\cdot\mathbf{d}_i}{\|\mathbf{q}\|\|\mathbf{d}_i\|}.
\]
The top-$k$ documents $\mathcal{D}_k \subseteq \mathcal{D}$ with the highest similarity scores are selected and concatenated with $Q$ as additional context for the generator $G(\cdot)$, which produces the final response:
\[
R = G(Q, \mathcal{D}_k).
\]
This retrieval-generation pipeline enables LLMs to access up-to-date information without retraining, improving factual accuracy and reducing hallucination~\cite{izacard2020leveraging,gao2023rag}.

\subsection{Zero-Knowledge Proofs}
\label{subsec:zkp}
A zero-knowledge proof (ZKP)~\cite{goldwasser1989zkp} is a protocol that allows a prover $P$ to convince a verifier $V$ that a statement $x$ is true, i.e., there exists a witness $w$ such that $R(x,w)=1$, without revealing any information about $w$ beyond the validity of the statement.

Modern succinct constructions such as Groth16~\cite{groth2016} and PLONK~\cite{gabizon2019plonk}, along with their subsequent developments, enable short proofs and fast verification, and are widely used for privacy-preserving and verifiable computation.

\subsection{Secure Multi-Party Computation}
\label{subsec:mpc}

Secure Multi-Party Computation (MPC) allows multiple parties to jointly compute a function over their private inputs without revealing those inputs to one another. Each party holds its own input and participates in a sequence of interactive computation rounds; at no point during these rounds does any party observe another party's raw input or any intermediate value in the clear. The final result is produced jointly, so that no single party ever reconstructs it alone.

In our framework, MPC protects the aggregation stage after zero-knowledge voting: committee nodes' hidden score components are combined jointly, revealing only the final aggregate while each node's contribution stays hidden.

\section{Threat Model and Trust Assumptions}
\label{sec:threat-model}

This section formalizes the threat model and security objectives RAG framework. We explicitly characterize the system participants, adversarial capabilities, and the information that must remain confidential in order to guarantee robustness, privacy, and security.

\subsection{System Model}

The system consists of five main components: (i) users who issue queries and receive generated responses, (ii) a RAG Service responsible for candidate retrieval and response generation, (iii) multiple Data Chains that maintain document provenance and finalized credibility metadata together with the corresponding scoring records, (iv) per-chain MPC committees that aggregate hidden scores, and (v) an Aggregator that records cross-chain binding hashes but does not produce a global recursive proof.

The RAG Service is not assumed to be fully trusted. Instead, its behavior is constrained through cryptographic commitments, chain-local verifiable state, cross-chain binding hashes, and a deterministic replay rule for ranking.

\subsection{Adversary Model}

We consider a probabilistic polynomial-time adversary with the
following capabilities:
\begin{itemize}
  \item observing all on-chain data and public network
  communication;
  \item injecting malicious documents or metadata into Data
  Chains;
  \item controlling a subset of voters or validators and
  attempting collusion or Sybil attacks;
  \item corrupting a subset of committee nodes involved in
  secure aggregation; and
  \item deviating from the prescribed protocol execution.
\end{itemize}

We assume that, within each consensus domain, an honest majority or
honest threshold of participants is maintained. For committee
aggregation, privacy and correctness follow the threshold assumption
of the underlying $t$-of-$n$ Shamir-sharing protocol.

To avoid ambiguity, we state three trust assumptions explicitly.

\begin{itemize}
  \item \textbf{Scores $\neq$ Truth.} A document's
  credibility score reflects committee endorsement, not truth. A
  high score means validators approved the document through the
  scoring protocol; it does not mean the content is factually
  correct or safe.

  \item \textbf{Completeness w.r.t.\ a Declared Set.}
  Completeness is defined relative to a declared chain set. Before
  each retrieval round $\mathsf{rid}$, the service and verifier
  agree on an ordered set of participating chains
  $\mathcal{J}_{\mathsf{rid}}$. A verifier can only check for
  omitted or reordered chains within this declared set, not against
  chains outside it.

  \item \textbf{Detectability, Not Availability.} Our
  guarantees cover tampering detection, not service availability.
  If the RAG Service alters data or deviates from the protocol,
  this is publicly detectable. However, we do not guarantee that
  the service will return a result at all; it may simply refuse to
  respond.
\end{itemize}

\subsection{Security and Privacy Objectives}

Our design aims to satisfy the following security and privacy
objectives.
\begin{itemize}
    \item \textbf{Security Property 1 (Data Integrity).} Once a
    chain-local tally is accepted, it must match the commitments
    recorded on-chain; results cannot be altered after the fact.

    \item \textbf{Security Property 2 (Vote Privacy).} Beyond the
    final aggregate, no one should learn an honest validator's
    score or identity.

    \item \textbf{Security Property 3 (Trust-Metadata Integrity).}
    A document's credibility score can only rise through valid,
    accepted votes. An adversary cannot inflate it without either
    corrupting enough of the trust domain or breaking the
    underlying cryptographic assumptions.

    \item \textbf{Security Property 4 (Retrieval Verifiability).}
    Clients can independently verify that the credibility scores
    and final ranking come from finalized chain state, computed by
    the declared deterministic rule---not fabricated by the
    service.
\end{itemize}

\subsection{Confidentiality and Public Verifiability}
To meet the objectives above, some information must stay hidden
during the protocol, while other information must be public enough
for anyone to check that nothing was tampered with.

Individual scores and voter identities must be hidden. If a score
were public, a validator could be bribed or threatened into voting
a certain way. If a voter's identity were linked to their score,
they could face retaliation or be profiled based on their voting
history. Hiding both prevents coercion, vote-buying, and collusion.

User queries, and anything that could reveal what a user was
looking for, are also kept off-chain. We do not offer a dedicated
query-hiding protocol beyond this; queries simply never touch the
chain. Likewise, the retrieval process itself is not shown to
outside observers in full: candidate document sets are hashed into
$l_j$, and only this hash is published. In short, the
only public information consists of cryptographic commitments,
chain-level tally records, zero-knowledge proofs, and cross-chain
binding hashes---nothing more.

Even with all this hidden, the system remains fully verifiable.
Each vote is published as a commitment plus a zero-knowledge proof
showing the vote was valid (from a registered validator, not a
duplicate, within the allowed score range) without revealing the
score itself. Each chain publishes its aggregate results: the
credibility scores for its documents, along with a digest binding
these scores to the corresponding document identities. The
Aggregator publishes, for each chain, a hash combining that chain's
score digest with its candidate-list digest, and then combines all
of these per-chain hashes into a single global digest. The service
similarly publishes hashes of the final ranking and the response.
Together, these let anyone verify that each chain's tally is
correct, no chain was skipped, and the final ranking is
consistent---all without ever seeing an individual vote.

\subsection{Out of Scope}
Our system does not guarantee that the LLM's output is factually
correct, and it does not defend against prompt injection or other
attacks on generation itself. It also cannot help if a validation
domain is fully compromised---for example, if enough validators and
committee members collude to endorse malicious content---or if a
validator simply makes a poor judgment call when scoring a document.

What we do guarantee is detectability, not prevention: as long as
the honest-threshold assumption holds, any attempt to tamper with
document provenance, forge aggregation results, omit required
chains, or manipulate trust scores can be publicly caught and
audited.

\section{System Overview}
\label{sec:system-overview}

\subsection{Architecture}

\begin{figure*}[t]
\centering
\includegraphics[width=\textwidth]{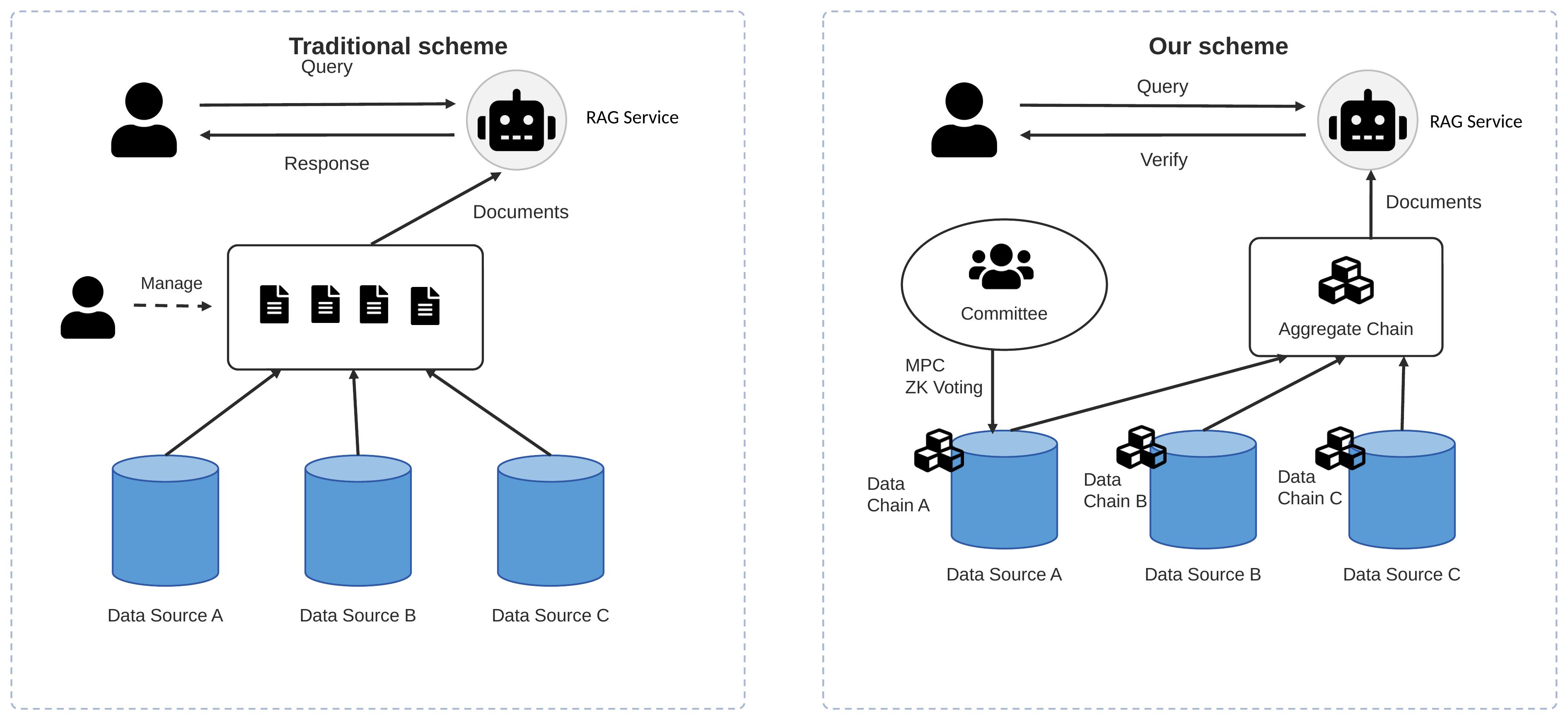}
\caption{System Overview}
\label{fig:system}
\end{figure*}

As shown in Fig.~\ref{fig:system}, the system has three main layers: \textbf{Data Chains}, the \textbf{Aggregator Binding Layer}, and the \textbf{RAG Service Layer}. Each Data Chain has its own MPC committee for secure tallying and optional auditing. The retrieval service itself is treated as untrusted and not necessarily decentralized. Document provenance is fixed at registration, trust scores are produced through protected scoring and aggregation, and cross-chain consistency is enforced via hash binding and deterministic replay, without relying on recursive proofs. In our prototype, the expensive cryptographic scoring is done once at ingestion or score-submission time, and the resulting artifacts are simply reused during later retrieval.

\subsubsection{Data Chains}
Each data chain manages its own corpus and validation logic independently. When a document is ingested, the chain records an immutable provenance anchor -- its identifier and hash -- along with the trust artifacts produced by validator scoring. When a query arrives, the retrieval service selects a set of candidate documents from the chain, and the chain exposes the finalized credibility scores for those candidates. The service then computes a hash over the exact candidate list, binding it to the scores it received. These records serve as verifiable attestations of chain-local credibility and form the basis for ranking. In our prototype, the zero-knowledge score proofs are generated once when scores are submitted, and reused across later queries rather than regenerated each time.

Score consolidation itself relies on committee-assisted MPC. Instead of revealing raw score components to a single aggregator, validators split their hidden score values into shares distributed across committee nodes, which jointly compute the sum and authorize its release. Only the aggregate document-level score is derived and recorded on-chain, together with artifacts needed for later verification; individual validator scores stay hidden. Committee-side threshold authorization and secret-sharing details mainly come into play during score finalization or post-hoc auditing.

\subsubsection{Aggregator Layer}
The Aggregator is the cross-chain binding layer. Instead of verifying one big recursive proof, it hashes each chain's score hash and candidate-list hash together into a single digest per chain, then combines all chains' digests into one global digest. This binding step prevents a chain from being skipped, a document's score from being swapped, or a candidate list from being replaced -- while keeping the cross-chain logic simple and lightweight.

MPC only plays a role earlier, inside each chain, as a privacy-preserving step for consolidating scores. By the time results reach the Aggregator, each Data Chain has already verified its own aggregated scores and finalized its per-document credibility values. The Aggregator simply binds each chain's finalized score hash and candidate-list hash together -- it does not repeat any of that chain-local verification.

\subsubsection{RAG Service Layer}
The RAG Service Layer is the off-chain interface between users and the verifiable multi-chain knowledge base. Given a query, it runs semantic retrieval to get candidate documents from each chain, fetches their credibility scores and proof references from the finalized chain state, and pulls the binding digest from the Aggregator. A trust-aware ranking step then combines semantic relevance with document credibility, so results reflect both how relevant a document is and how trustworthy its source is. The top-ranked documents are passed to the LLM generator, and the final output is bundled into a ProofPack containing the metadata a client needs to verify it.

Together, these three layers form a unified architecture: data quality assessment is decentralized across chains, integrity is enforced through binding hashes, and the off-chain retrieval service itself stays accountable and verifiable.

\subsection{Query Flow}
\label{subsec:queryflow}

\begin{table}[t]
\centering
\caption{Notation used in the query flow.}
\label{tab:notation}
\begin{tabular}{cl}
\toprule
Symbol & Meaning \\
\midrule
$Q$              & User query \\
$\mathcal{C}$    & Set of chains participating in the query \\
$D_j$            & Candidate document set from chain $j$ \\
$\mathrm{docId}$ & Document identifier \\
$w_j$            & Document credibility score on chain $j$ \\
$h_j$            & Binding digest for chain $j$ (score + candidate list) \\
$G$              & Global digest across all chains \\
$D^\star$        & Final evidence set passed to the generator \\
$R$              & Generated response \\
$\Pi$            & Proof package returned to the client \\
\bottomrule
\end{tabular}
\end{table}

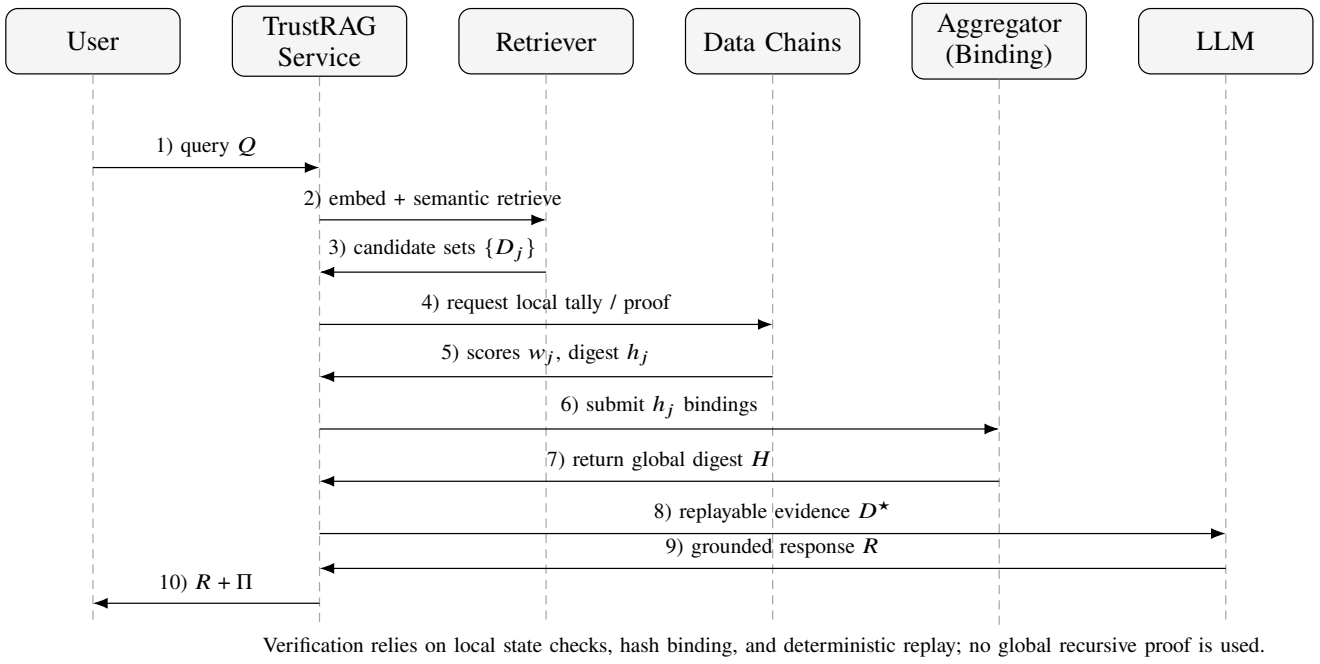
\begin{figure*}[t]
    \centering
    \resizebox{0.97\textwidth}{!}{
    \begin{tikzpicture}[
        font=\small,
        >=Latex,
        actor/.style={draw, rounded corners, minimum width=2.0cm, minimum height=0.75cm, align=center, fill=gray!8},
        lifeline/.style={densely dashed, gray!70},
        msg/.style={font=\scriptsize}
    ]
    \node[actor] (u) at (0,0) {User};
    \node[actor] (s) at (2.6,0) {TrustRAG\\Service};
    \node[actor] (r) at (5.2,0) {Retriever};
    \node[actor] (d) at (7.8,0) {Data Chains};
    \node[actor] (a) at (10.4,0) {Aggregator\\(Binding)};
    \node[actor] (l) at (13.0,0) {LLM};
    \draw[lifeline] (u.south) -- (0,-6.7);
    \draw[lifeline] (s.south) -- (2.6,-6.7);
    \draw[lifeline] (r.south) -- (5.2,-6.7);
    \draw[lifeline] (d.south) -- (7.8,-6.7);
    \draw[lifeline] (a.south) -- (10.4,-6.7);
    \draw[lifeline] (l.south) -- (13.0,-6.7);
    \draw[->] (0,-1.45) -- node[above,msg] {1) query $Q$} (2.6,-1.45);
    \draw[->] (2.6,-2.05) -- node[above,msg] {2) embed + semantic retrieve} (5.2,-2.05);
    \draw[->] (5.2,-2.65) -- node[above,msg] {3) candidate sets $\{D_j\}$} (2.6,-2.65);
    \draw[->] (2.6,-3.25) -- node[above,msg] {4) request local tally / proof} (7.8,-3.25);
    \draw[->] (7.8,-3.85) -- node[above,msg] {5) scores $w_j$, digest $h_j$} (2.6,-3.85);
    \draw[->] (2.6,-4.45) -- node[above,msg] {6) submit $h_j$ bindings} (10.4,-4.45);
    \draw[->] (10.4,-5.05) -- node[above,msg] {7) return global digest $H$} (2.6,-5.05);
    \draw[->] (2.6,-5.65) -- node[above,msg] {8) replayable evidence $D^\star$} (13.0,-5.65);
    \draw[->] (13.0,-6.05) -- node[above,msg] {9) grounded response $R$} (2.6,-6.05);
    \draw[->] (2.6,-6.45) -- node[above,msg] {10) $R + \Pi$} (0,-6.45);
    \node[font=\scriptsize, align=center] at (7.7,-6.95)
    {Verification relies on local state checks, hash binding, and deterministic replay; no global recursive proof is used.};
    \end{tikzpicture}
    }
\caption{End-to-end query flow. Candidate sets are matched to local trust artifacts, then bound across chains through hashes and verified by replay.}
    \label{fig:queryflow}
\end{figure*}

Figure~\ref{fig:queryflow} shows how a query flows through the system: from the RAG Service, to the Data Chains, through the Aggregator binding layer, and finally to the LLM, with committee-side MPC aggregation happening inside each chain's trust path. The key design choice is to separate document registration and trust scoring from online retrieval: documents are registered and scored ahead of time, and committee authorization is only invoked when scores are finalized or audited. This means that during normal query processing, the service just consumes already-finalized credibility scores for the retrieved candidates -- it does not redo any scoring. Global coordination at query time is then reduced to two simple things: binding hashes together and replaying the ranking over the declared set of chains $\mathcal{C}$.

\paragraph*{Steps 1--3: Query submission, embedding, and candidate-set binding.}
When the RAG Service receives a query $Q$, it first encodes it into a vector $\mathbf{q} = f(Q)$ and retrieves a candidate document set $D_j$ from each participating chain $j \in \mathcal{C}$. For each chain, the service also computes a list hash $l_j = H(D_j)$, which binds the exact set of candidates used in this round. This hash is computed fresh at query time -- it is not fixed in advance, since the candidate set depends on the specific query.

\paragraph*{Steps 4--6: Chain-local trust retrieval, optional tallying, and local checks.}
For each candidate document $d_t \in D_j$, the service retrieves its finalized credibility score from chain $j$. In the full design, validators on chain $j$ submit hidden scores together with zero-knowledge proofs of membership, uniqueness, range validity, and commitment correctness. Committee nodes then aggregate these hidden scores via MPC to derive each document's credibility $w_t$, without revealing any individual validator's score.

Let $W_j = \{(\text{docId}_t, w_t)\}_{d_t \in D_j}$ denote the set of scores for the retrieved candidates on chain $j$, and $s_j = H(W_j)$ its hash. In our prototype, scores are computed once when validators submit them, and simply reused at query time. So at query time, the service just fetches the relevant subset $W_j$, computes $s_j$, and checks it against the chain's on-chain records. In short: score computation happens ahead of time, while candidate-set binding happens at query time.

\paragraph*{Steps 7--8: Cross-chain binding and deterministic reranking.}
Once the chain-local outputs are ready, the Aggregator combines them into a single digest:
\[
h_j = H(s_j \,\|\, l_j), \qquad
G = H(h_1 \,\|\, h_2 \,\|\, \cdots \,\|\, h_m),
\]
where the chains are taken in a fixed canonical order over $\mathcal{C}$. This makes it detectable if any chain is dropped or reordered relative to the declared configuration.
The RAG Service then ranks each candidate document by combining semantic relevance with credibility:
\[
\mathrm{score}(d_i) = \alpha \cdot \mathrm{sim}(Q, d_i) + \beta \cdot w_i.
\]
The top-ranked documents form the final evidence set $D^\star$, which is bound into a replay hash:
\[
r = H\!\big(Q \,\|\, \mathcal{C} \,\|\, \{s_j\}_j \,\|\, \{l_j\}_j \,\|\, D^\star\big).
\]
$D^\star$ is then passed to the LLM, which produces the final grounded response:
\[
R = G(Q, D^\star).
\]

\paragraph*{Steps 9--10: Returning a verifiable response package.}
The response is bundled with a verifiable proof package:
\[
\Pi = \big(\mathcal{C},\ \{W_j\}_j,\ \{s_j\}_j,\ \{l_j\}_j,\ G,\ r,\ H(R)\big),
\]
while the underlying chain-local records remain retrievable from their source chains. The response and $\Pi$ are returned to the user, who can independently verify chain-local correctness, cross-chain completeness, ranking correctness, and response integrity.

This query flow provides end-to-end verifiability without a global recursive proof: correctness is instead established by checking chain-local state, binding hashes across chains, and replaying the ranking computation.

\section{Protocol Design}
\label{sec:protocol-design}

\subsection{Document Registration}

Before the system handles any query, each Data Chain registers its documents up front. For a document $d$, registration just needs a content hash, an embedding hash, and a metadata hash. The chain stores this as a tuple
\[
(\text{docId},\ \text{contentHash},\ \text{embeddingHash},\ \text{owner}),
\]
which becomes the document's permanent identity for everything that happens later -- retrieval, scoring, verification, all of it.

Registration itself doesn't produce a trust score. Its only job is to fix a stable identity for the document on-chain, so that when scores and verification happen later, they're always anchored to something concrete rather than to raw content that could shift or be reinterpreted.

\subsection{Zero-Knowledge Quality-Scoring Protocol}
\label{subsec:zk-quality-scoring}

RAG systems depend on the quality of their documents. Public expert ratings, however, expose validators to bribery, coercion, and retaliation. We therefore require a scoring protocol that is verifiable without being public.

The zero-knowledge scoring protocol operates over the candidate set $D_j$ on chain $j$ for a given round. It enforces three properties simultaneously: only registered validators may vote; each validator may score a given document at most once per round; and every accepted score lies in $[0, S_{\max}]$. None of these properties requires disclosing the voter's identity or vote value.

The protocol's output is a document-level weight $w_t$, bound to a stable document identifier. At query time, the service discloses the weights for retrieved documents and binds them into a score hash $s_j$ together with the candidate-list hash $l_j$; no rescoring occurs on the query path.

Once a vote is proven valid, an MPC layer sums the hidden scores (Section~\ref{subsec:agg-circuit}). This does not replace the zero-knowledge layer; it ensures the intermediate values used in tallying are never exposed to a single party in plaintext.

\subsubsection{Participants}

The protocol involves three roles.

\begin{itemize}
    \item \textbf{Validators} are domain experts who score documents. Each holds a secret key $sk_i$ and is listed in a registry, proving membership via a Merkle tree rather than revealing identity.
    \item \textbf{Data Chains} govern a domain corpus. They verify proofs, maintain tally records, and reject duplicate or malformed submissions, without ever observing a score or validator identity in plaintext.
    \item \textbf{Committee nodes} jointly aggregate hidden scores under a threshold scheme. The prototype uses $n=4$ nodes with reconstruction threshold $t=3$.
\end{itemize}

\subsubsection{Cryptographic Building Blocks}

Three primitives support the protocol.

\textbf{Pedersen commitments} hide the score. A validator commits to score $s_{i,t}$ with blinding factor $r_{i,t}$:
\[
\mathsf{Com}_{i,t} = g^{s_{i,t}} h^{r_{i,t}} \bmod p.
\]
The commitment reveals nothing about $s_{i,t}$ and cannot later be altered. Its multiplicative structure allows many commitments to be combined into a sum without opening individual scores, which is essential for tallying.

\textbf{Merkle membership} proves registration without identifying the validator. All validators are committed into a tree with root $\mathsf{Root}_{reg}$; a validator supplies an authentication path as part of the proof.

\textbf{Nullifiers} prevent duplicate voting. Each validator computes
\[
\mathsf{Null}_{i,t} = \mathsf{Poseidon}(sk_i, \mathsf{rid}, \text{docId}_t),
\]
bound to their key, the round $\mathsf{rid}$, and the document. The nullifier is published on-chain, so a repeated vote on the same document in the same round is detected, without revealing which validator cast it.

Each vote additionally carries a succinct zk-SNARK (Groth16 or PLONK) attesting that all constraints above hold, at constant proof size and low verification cost.

\subsubsection{Formal Relation}

For validator $v_i$ scoring document $\text{docId}_t$ on chain $j$, the public statement is
\[
x_{i,t} = (\mathsf{Root}_{reg},\ \mathsf{Com}_{i,t},\ \mathsf{Null}_{i,t},\ \mathsf{rid},\ \text{docId}_t,\ S_{\max}),
\]
and the private witness is
\[
w_{i,t} = (s_{i,t},\ r_{i,t},\ sk_i,\ \mathsf{MerklePath}_i).
\]

The relation $\mathcal{R}_{\mathrm{score}}$ holds iff the following four constraints are satisfied:
\begin{align}
\mathsf{Root}_{reg} &= \mathsf{MerkleRoot}(\ell_i, \mathsf{MerklePath}_i) \label{eq:membership}\\
\mathsf{Null}_{i,t} &= \mathsf{Poseidon}(sk_i, \mathsf{rid}, \text{docId}_t) \label{eq:uniqueness}\\
\mathsf{Com}_{i,t} &= g^{s_{i,t}} h^{r_{i,t}} \bmod p \label{eq:commitment}\\
0 &\leq s_{i,t} \leq S_{\max} \label{eq:range}
\end{align}
Constraint~\eqref{eq:membership} enforces membership, \eqref{eq:uniqueness} enforces vote uniqueness, \eqref{eq:commitment} enforces commitment correctness, and \eqref{eq:range} enforces range validity.

The validator generates $\pi_{i,t} \leftarrow \mathsf{Prove}(pk_{\mathrm{score}}, x_{i,t}, w_{i,t})$ and submits $(\mathsf{Com}_{i,t}, \mathsf{Null}_{i,t}, \pi_{i,t})$. The Data Chain accepts the submission only if $\mathsf{Verify}(vk_{\mathrm{score}}, x_{i,t}, \pi_{i,t}) = 1$.

\subsection{Credibility and Cross-Chain Binding}
\label{subsec:agg-circuit}

Once each chain has finalized its document-level scores, the protocol discloses the subset relevant to the current candidate set and binds these per-chain outputs together, without a recursive global proof.

The chain-local inputs to this step come from secure MPC aggregation: committee-side secure sum strengthens the local tally, while cross-chain coordination reduces to hash binding over the published outputs. The resulting aggregate is checked against the accumulated Pedersen commitment product before the chain derives per-document scores and publishes the records later consumed at retrieval.

\subsubsection{Chain-Local Tally and Weight Computation}

For each chain $j \in \mathcal{C}$, let
\[
l_j = H(D_j)
\]
bind the candidate set for the current round. For each document $d_t \in D_j$, the committee computes the hidden tally
\[
S_t^{(j)} = \sum_i s_{i,t}^{(j)},
\]
and the aggregated commitment
\[
\mathsf{ComTally}_t^{(j)} = \prod_i \mathsf{Com}_{i,t}^{(j)} = g^{S_t^{(j)}} h^{R_t^{(j)}} \bmod p.
\]
The credibility value for each candidate is
\[
w_t^{(j)} = \frac{S_t^{(j)}}{n_j \cdot S_{\max}},
\]
where $n_j$ is the validator count for chain $j$. This keeps $w_t^{(j)} \in [0,1]$ and makes weights comparable across documents and chains. The chain then forms
\[
\mathcal{W}_j = \{(\text{docId}_t, w_t^{(j)})\}_{d_t \in D_j}, \qquad s_j = H(\mathcal{W}_j).
\]

\subsubsection{Chain-Local Verifiable Relation}

Each chain exposes enough finalized state for a verifier to check its local tallies and derived weights. In the prototype, this consists of finalized tally records, aggregate openings, and contract-verifiable checks. The relation $\mathcal{R}_{\mathrm{all}}^{(j)}$ holds iff:
\begin{align}
&\text{all accepted votes satisfy } \mathcal{R}_{\mathrm{score}}, \label{eq:votes-valid}\\
&\mathsf{ComTally}_t^{(j)} = \textstyle\prod_i \mathsf{Com}_{i,t}^{(j)}, \quad \forall d_t \in D_j, \label{eq:tally-commit}\\
&S_t^{(j)} = \textstyle\sum_i s_{i,t}^{(j)}, \quad \forall d_t \in D_j, \label{eq:tally-sum}\\
&w_t^{(j)} = \frac{S_t^{(j)}}{n_j \cdot S_{\max}}, \quad \forall d_t \in D_j, \label{eq:weight-def}\\
&s_j = H(\mathcal{W}_j). \label{eq:scorehash-def}
\end{align}
Thus the disclosed chain-local state certifies both the validity of the underlying votes and the correctness of the derived weights, for the exact candidate set hashed into $l_j$.

\subsubsection{Binding Layer and RAG Integration}

The Aggregator stores only binding hashes:
\begin{align*}
h_j &= H(s_j \,\|\, l_j), \\
G &= H(h_1 \,\|\, h_2 \,\|\, \cdots \,\|\, h_m).
\end{align*}
At retrieval, the service obtains $\{\mathcal{W}_j\}_j$, $\{s_j\}_j$, and $\{l_j\}_j$ from the Data Chains, and $G$ from the Aggregator. Any client can then verify the supporting chain-local state, recompute $h_j$ and $G$, and replay the ranking procedure. The Aggregator thus acts only as a binding layer, preventing substitution or omission relative to the declared chain set $\mathcal{C}$, without introducing a recursive proof.

\subsection{Trust-Aware Retrieval and Verifiable Response Generation}
\label{subsec:trust-aware-rag}

Once each chain has published $(\mathcal{W}_j, s_j, l_j)$ along with its supporting on-chain state, the RAG Service Layer performs provenance-aware retrieval and generates a grounded, verifiable response. Unlike conventional RAG, which treats all retrieved documents as equally trustworthy, our system folds cryptographically verified credibility into ranking: semantic relevance and provenance jointly determine which documents shape the output.

Algorithm~\ref{alg:trust-aware-rag} gives the steps.

\begin{algorithm}[t]
\caption{Trust Retrieval and Verifiable Response Generation}
\label{alg:trust-aware-rag}
\begin{algorithmic}[1]
\Require Query $Q$, round id $\mathsf{rid}$, chain set $\mathcal{C}$, proof interfaces for all Data Chains, binding digest $G$
\Ensure Response $R$ and proof package $\Pi$
\State $\mathbf{q} \gets f(Q)$
\State $\{D_j\}_{j \in \mathcal{C}} \gets \textsc{SemanticRetrieveByChain}(\mathbf{q})$
\ForAll{$j \in \mathcal{C}$}
    \State $l_j \gets H(D_j)$
    \State $(\mathcal{W}_j, s_j) \gets \textsc{GetLocalState}(j, \mathsf{rid}, D_j)$
    \If{$\textsc{VerifyLocalState}(j, \mathsf{rid}, l_j, s_j, \mathcal{W}_j) = 0$}
        \State \textbf{abort} with \textsc{InvalidLocalState}
    \EndIf
    \State $h_j \gets H(s_j \,\|\, l_j)$
\EndFor
\If{$H(h_1 \,\|\, \cdots \,\|\, h_m) \neq G$}
    \State \textbf{abort} with \textsc{BindingMismatch}
\EndIf
\ForAll{$d_i^{(j)} \in \bigcup_j D_j$}
    \State $\mathrm{sim}_i \gets \mathrm{sim}(Q, d_i^{(j)})$
    \State $\mathrm{score}_i \gets \alpha \cdot \mathrm{sim}_i + \beta \cdot w_i^{(j)}$
\EndFor
\State $D^{\star} \gets \textsc{Top-}k(\mathrm{score}_i)$
\State $\mathsf{ctx} \gets Q \,\|\, \mathcal{C} \,\|\, \{s_j\}_j \,\|\, \{l_j\}_j \,\|\, D^{\star}$
\State $r \gets H(\mathsf{ctx})$
\State $R \gets G(Q, D^{\star})$
\State $h_R \gets H(R)$
\State $\Pi \gets \big(\mathsf{rid},\ \mathcal{C},\ \{\mathcal{W}_j\}_j,\ \{s_j\}_j,$
\Statex \hspace{4em} $\{l_j\}_j,\ G,\ r,\ h_R\big)$
\State \textsc{EmitResponse}$(R, \Pi)$
\end{algorithmic}
\end{algorithm}

Each returned answer is thus conditioned only on documents whose provenance and credibility were independently verified, and comes with a replayable certificate any client can check.

\section{Security and Privacy Analysis}
\label{sec:security-analysis}

We analyze the security and privacy of TrustRAG, showing that the protocol achieves data integrity, vote privacy, trust-metadata integrity, retrieval verifiability, chain-level consistency, vote uniqueness, cross-chain binding integrity, retrieval atomicity, and liveness, under standard cryptographic assumptions.

\subsection{Assumptions}

\begin{itemize}
    \item \textbf{A1. Cryptographic hardness.} The hash function is collision resistant, the commitment scheme is computationally binding and hiding, and the zero-knowledge proof system is complete, sound, and zero-knowledge.
    \item \textbf{A2. Honest-threshold trust domains.} In each chain, the fraction of corrupted validators and committee nodes stays below the threshold required by the local voting and aggregation procedures.
    \item \textbf{A3. Declared chain set and deterministic ordering.} For each query, all honest parties agree on the ordered chain set $\mathcal{C}$ against which completeness and binding are evaluated.
    \item \textbf{A4. Deterministic verification logic.} Given the same public input, all honest contracts and verifiers return the same result for vote proofs, chain-local state checks, binding checks, and replay checks.
    \item \textbf{A5. Eventual delivery.} Messages among honest parties and access to required chain state are eventually available.
\end{itemize}

\subsection{Proofs}

\begin{theorem}[Data Integrity]
For any document $\text{docId}_t \in D_j$, if a Data Chain accepts an aggregate opening $(S_t,R_t)$ and its tally commitment $\mathsf{ComTally}_t$, the accepted tally is consistent with the exact set of vote commitments recorded on-chain for $\text{docId}_t$.
\end{theorem}

\textit{Proof.} Each vote commitment is recorded on-chain only after its zero-knowledge proof is verified. The submitted aggregate is accepted only if the chain verifies the Pedersen consistency equation $g^{S_t}h^{R_t} = \prod_i \mathsf{Com}_{i,t}$ together with vote-count and authorization checks. To make the chain accept a different tally, an adversary would need to alter the finalized commitment set, produce a second valid opening for the same commitment product, or bypass verification -- contradicting chain immutability, computational binding, and soundness, respectively. \hfill $\square$

\begin{theorem}[Vote Privacy]
For any honest validator, the adversary learns nothing about the private vote value $s_{i,t}$ beyond what the public protocol output reveals.
\end{theorem}

\textit{Proof.} Each validator publishes only a commitment, a nullifier, and a zero-knowledge proof $\pi_{\text{vote}}$. The commitment is computationally hiding, so it conceals $(s_{i,t}, r_{i,t})$; the proof system is zero-knowledge, so $\pi_{\text{vote}}$ reveals nothing beyond the statement's truth. \hfill $\square$

\begin{theorem}[Trust-Metadata Integrity Under Poisoning]
An adversary cannot raise a document's finalized credibility score beyond what accepted votes imply, unless it corrupts the trust domain beyond the tolerated threshold or breaks the protocol's cryptographic assumptions.
\end{theorem}

\textit{Proof.} Inflating a score requires either injecting favorable votes or manipulating the tally or binding outcome. The first is blocked by vote validity: only registered validators can vote, nullifiers prevent duplicate votes, commitment correctness binds each vote to a concrete score, and range checks reject out-of-domain values. The second is blocked because a local aggregate is accepted only if it matches the accepted commitments, and $W_j$ is accepted only if consistent with the finalized tally for the candidate set hashed into $l_j$ and $s_j$; the global digest $G$ then prevents substituting $W_j$ or the candidate list undetected. Thus credibility can rise only through accepted votes, threshold corruption, or broken assumptions -- this does not claim honest validators always score poisoned content correctly, only that published metadata faithfully reflects the accepted scoring process. \hfill $\square$

\begin{theorem}[Retrieval Verifiability]
For any response returned with $\Pi = (\mathcal{C}, \{W_j\}_j, \{s_j\}_j, \{l_j\}_j, G, r, H(R))$, any client can efficiently verify that the credibility values and final ranking originated from accepted chain state.
\end{theorem}

\textit{Proof.} The client fetches finalized chain state for all $j \in \mathcal{C}$ and checks that each published $(W_j, s_j, l_j)$ is consistent with accepted local commitments and tallies (soundness of vote proofs, A4). It recomputes $h_j = H(s_j \,\|\, l_j)$ and $G' = H(h_1 \,\|\, \cdots \,\|\, h_m)$, checking $G' = G$. Using the disclosed candidates, chain set, and weights, it recomputes ranking scores, replays the top-$k$ selection, checks the result against $r$, and verifies the response hash. If the service tampers with scores, substitutes candidates, omits a chain, reranks incorrectly, or alters the response, one of these checks fails. \hfill $\square$

\begin{theorem}[Chain-Level Consistency]
If an honest party accepts a finalized state $S$ for a round, no other honest party accepts a conflicting state $S' \neq S$ for that round.
\end{theorem}

\textit{Proof.} Local acceptance requires a verified proof and commitment tuple; since verification is deterministic and sound, two conflicting local states cannot both be accepted. At the binding layer, $G$ is a deterministic function of the ordered $h_j$'s, so a conflicting global binding would require forging a local proof, finding a hash collision, or violating finality. \hfill $\square$

\begin{theorem}[Local Vote Validity and Uniqueness]
If a commitment $\mathsf{Com}_{i,t}$ is accepted for document $\text{docId}_t$, it came from a registered validator, is bound to a score within range, and that validator cannot cast a second valid vote for the same document.
\end{theorem}

\textit{Proof.} The vote circuit enforces membership, uniqueness (via nullifier $\mathsf{Null}_{i,t} = \mathsf{Poseidon}(sk_i,\text{docId}_t)$), commitment correctness, and range validity simultaneously. By soundness, any accepted proof implies all four hold; nullifier reuse blocks double voting, and the range check blocks malformed scores. \hfill $\square$

\begin{theorem}[Cross-Chain Binding Integrity]
If the Aggregator publishes $G$ over the declared chain set $\mathcal{C}$, every $(s_j, l_j)$ pair is bound to a unique $h_j$, and any change to a score, candidate set, or chain membership is detectable.
\end{theorem}

\textit{Proof.} The Aggregator computes $h_j = H(s_j \,\|\, l_j)$ for each chain in $\mathcal{C}$ and $G$ over that ordered sequence. Replacing a score or candidate-list hash changes $h_j$ except with negligible probability under collision resistance; omitting or reordering a chain likewise changes $G$. \hfill $\square$

\begin{theorem}[Retrieval Atomicity]
For any query $Q$, the response $R$ is generated only from a fully verified evidence set $D^\star$ with a valid $\Pi$, or no response is accepted.
\end{theorem}

\textit{Proof.} The service does not finalize generation on partial outputs; it waits until all required $(W_j, s_j, l_j)$ tuples are obtained and $G$ is published. Only after verifying local state and the global binding does it derive the ranking, compute $r$, and invoke the generator. If any required result is invalid, missing, or inconsistent, local verification, the binding check, or the replay check fails, and no valid $\Pi$ can be formed. \hfill $\square$

\begin{theorem}[Liveness]
Assume corrupted validators stay below threshold in each chain and messages between honest parties are eventually delivered. Then every valid query eventually yields either a finalized response with a valid $\Pi$, or an explicit rejection from proof failure, binding mismatch, or insufficient evidence.
\end{theorem}

\textit{Proof.} A valid query is dispatched to $\mathcal{C}$. Under A2, A4, A5, each honest chain eventually exposes valid local state or a failure indication. These are relayed to the Aggregator, which deterministically computes $G$ once all inputs arrive in canonical order. The service then either receives valid chain-local state with a matching $G$ and completes generation, or detects failure and returns rejection. \hfill $\square$

\section{Implementation and Evaluation}
\label{sec:implementation-evaluation}

We implemented a prototype realizing the main components of the design and evaluate its overhead. The vote-validity and score-consolidation path uses a zero-knowledge circuit in \texttt{circom~2.1.6}, with on-chain verification via Solidity contracts on a Hardhat EDR simulated network. Committee-side aggregation uses Shamir secret sharing and an MP-SPDZ-compatible secure-sum interface, with threshold authorization performed off-chain before aggregate submission.\footnote{\url{https://github.com/1Vastsky/trustRAG}}

\subsection{Prototype Components}
The prototype consists of four components.
\begin{itemize}
    \item \textbf{Document and retrieval service.} An off-chain service stores document contents, maintains the retrieval index, derives per-chain candidate sets, and performs replayable reranking.
    \item \textbf{Validator-scoring path.} A zero-knowledge vote circuit enforces validator membership, nullifier uniqueness, commitment correctness, and score range validity before a vote is accepted.
    \item \textbf{Committee aggregation path.} Committee nodes receive Shamir shares of hidden score components, perform secure-sum aggregation, and authorize the final aggregate once the threshold is met.
    \item \textbf{On-chain verification and binding path.} Solidity contracts maintain commitment products, accepted vote counts, aggregate summaries, and the cross-chain binding digests $(h_j, G)$, rejecting any submission that fails the Pedersen consistency or authorization checks.
\end{itemize}

\subsection{Prototype Boundary}
The implementation directly exercises the security-critical chain-local tally path and the publication of replay metadata consumed by the service layer. Committee aggregation uses a $t$-of-$n$ Shamir-sharing domain with $n=4$, $t=3$. The prototype fully implements the vote-validity circuit, the contract-level commitment-consistency checks, the disclosure of finalized per-document weights, and the binding metadata $(s_j, l_j, G)$ used for replay -- sufficient to demonstrate the end-to-end vote-validation path, the commitment-consistency check, and retrieval-time use of authenticated replay metadata, without recursive proof aggregation.

\subsection{Experimental Setup}

We evaluate two parts of the system. First, we benchmark the zero-knowledge vote-validity circuit under Groth16 and PLONK. Second, we measure the cost of the committee-based chain-local tally path, including Pedersen commitments, Shamir share splitting and reconstruction, and threshold authorization. Unless otherwise stated, we use a score bound $S_{\max}=10$, a committee size of $4$, and threshold $3$.

All experiments were conducted on a MacBook Pro equipped with an Apple M1 Pro processor and 32GB of memory, running macOS. The zk modules are implemented with Circom and \texttt{snarkjs}, while the blockchain path is implemented with Solidity contracts and exercised through Hardhat. We evaluate the protocol in a modular manner: vote-proof generation, chain-local tallying, threshold authorization, and committee reconstruction are measured separately and then interpreted as composable building blocks of the full system.

\subsection{Proving Time}
\label{subsec:prove-time}

We first compare the proving latency of Groth16 and PLONK for the same $\mathcal{R}_{\mathsf{vote}}$ circuit under different Merkle depths. Table~\ref{tab:prove-time} summarizes the results. As expected, Groth16 achieves substantially lower proving time across all tested depths, while PLONK incurs a much larger overhead due to its universal constraint system and polynomial-commitment machinery.

\begin{table}[t]
    \centering
    \caption{Proving time comparison for the vote circuit under different Merkle depths.}
    \label{tab:prove-time}
    \begin{tabular}{ccc}
        \toprule
        \textbf{Depth} & \textbf{Groth16 (ms)} & \textbf{PLONK (ms)} \\
        \midrule
        8  & 865.73   & 5{,}183.92  \\
        12 & 1{,}023.82 & 9{,}930.46  \\
        16 & 1{,}486.58 & 9{,}998.97  \\
        20 & 1{,}624.51 & 9{,}946.89 \\
        \bottomrule
    \end{tabular}
\end{table}

Even at depth $20$, Groth16 remains in the low-second regime, which is acceptable because proof generation happens only at vote-submission time rather than on the online query path. In contrast, PLONK becomes significantly more expensive in this setting, suggesting that Groth16 is the more practical choice for the current prototype.

\subsection{Circuit Scaling with Merkle Depth}
\label{subsec:scaling-depth}

To evaluate scalability with respect to validator-set size, we vary the Merkle tree depth while keeping the rest of the vote circuit unchanged. Concretely, we instantiate \texttt{VoteCircuit(depth)} for depths ranging from $8$ to $50$, compile each variant to R1CS, and record the resulting circuit artifacts.

\begin{table*}[t]
    \centering
    \caption{Circuit artifact scaling with Merkle tree depth.}
    \label{tab:scaling-depth}
    \begin{tabular}{ccccccc}
        \toprule
        \textbf{Depth} &
        \textbf{\#Constraints} &
        \textbf{Nonlinear} &
        \textbf{Linear} &
        \textbf{Wires} &
        \textbf{R1CS size (KB)} &
        \textbf{WASM size (KB)} \\
        \midrule
        8  & 5{,}251  & 2{,}488  & 2{,}763  & 5{,}263  & 702  & 1{,}738 \\
        12 & 7{,}339  & 3{,}472  & 3{,}867  & 7{,}355  & 981  & 1{,}751 \\
        16 & 9{,}427  & 4{,}456  & 4{,}971  & 9{,}447  & 1{,}260 & 1{,}763 \\
        20 & 11{,}515 & 5{,}440  & 6{,}075  & 11{,}539 & 1{,}538 & 1{,}774 \\
        24 & 13{,}604 & 6{,}486  & 7{,}118  & 13{,}628 & 1{,}812 & 1{,}786 \\
        28 & 15{,}694 & 7{,}493  & 8{,}201  & 15{,}721 & 2{,}092 & 1{,}799 \\
        32 & 17{,}785 & 8{,}523  & 9{,}262  & 17{,}812 & 2{,}365 & 1{,}813 \\
        40 & 21{,}968 & 10{,}537 & 11{,}431 & 22{,}002 & 2{,}926 & 1{,}842 \\
        50 & 27{,}181 & 13{,}045 & 14{,}136 & 27{,}221 & 3{,}622 & 1{,}876 \\
        \bottomrule
    \end{tabular}
\end{table*}

The artifact sizes exhibit a stable near-linear growth trend as the Merkle depth increases. In particular, the R1CS size grows from 702 KB at depth $8$ to 3.6 MB at depth $50$, while the compiled WASM artifact remains within a relatively narrow range of approximately 1.7--1.9 MB. This suggests that the dominant scaling cost comes from the constraint system itself rather than from an explosion in executable representation size.

\subsection{Chain-Local Tally Overhead}
\label{subsec:registration-overhead}

We next measure the cost of the chain-local tally path under committee-based secure aggregation. Table~\ref{tab:registration-scale} reports end-to-end latency for 20, 50, and 100 validators. The results show that the dominant cost in the current prototype comes from threshold authorization, while Pedersen commitments and Shamir operations remain small compared with the total tally delay.

\begin{table}[t]
    \centering
    \caption{Chain-local tally overhead with committee size $4$ and threshold $3$.}
    \label{tab:registration-scale}
    \resizebox{\columnwidth}{!}{%
    \begin{tabular}{ccccc}
        \toprule
        \textbf{Votes} & \textbf{Latency (ms)} & \textbf{Commit (ms)} & \textbf{Split+Recon (ms)} & \textbf{BLS (ms)} \\
        \midrule
        20  & 1123.24 & 0.10 & 0.40 & 1122.66 \\
        50  & 1117.12 & 0.24 & 0.87 & 1115.87 \\
        100 & 1118.86 & 0.47 & 1.74 & 1116.39 \\
        \bottomrule
    \end{tabular}%
    }
\end{table}

This profile is consistent with the intended design: commitment and Shamir operations scale gently with the number of votes, while the dominant cost in the current prototype lies in finalizing the chain-local tally.

\subsection{Committee Scaling}
\label{subsec:committee-scaling}

To complement the fixed-size tally benchmark above, we also evaluate how the two committee-critical subroutines scale with committee size: BLS threshold-signature aggregation and Shamir reconstruction. Figure~\ref{fig:committee-scale} visualizes the measurements for committee sizes from $4$ to $48$ under three threshold ratios, $t/n \in \{50\%, 75\%, 100\%\}$.

\begin{figure*}[t]
    \centering
    \begin{minipage}{0.48\textwidth}
        \centering
        \begin{tikzpicture}
        \begin{axis}[
            width=\linewidth,
            height=0.62\linewidth,
            xlabel={Committee size $n$},
            ylabel={Aggregation cost ($\mu$s)},
            xmin=0, xmax=50,
            ymin=0,
            grid=both,
            legend style={at={(0.02,0.98)}, anchor=north west, font=\scriptsize},
            tick label style={font=\scriptsize},
            label style={font=\small},
        ]
        \addplot+[mark=*] coordinates {
            (4,6000) (6,9000) (8,12000) (12,18000) (16,24000)
            (20,30000) (24,36000) (28,42000) (32,48000)
            (40,60000) (48,72000)
        };
        \addlegendentry{$t/n=50\%$}

        \addplot+[mark=square*] coordinates {
            (4,12000) (6,18000) (8,24000) (12,36000) (16,48000)
            (20,60000) (24,72000) (28,84000) (32,96000)
            (40,120000) (48,160000)
        };
        \addlegendentry{$t/n=75\%$}

        \addplot+[mark=triangle*] coordinates {
            (4,24000) (6,36000) (8,48000) (12,72000) (16,96000)
            (20,120000) (24,144000) (28,168000) (32,192000)
            (40,240000) (48,288000)
        };
        \addlegendentry{$t/n=100\%$}
        \end{axis}
        \end{tikzpicture}

        \small BLS aggregation
    \end{minipage}\hfill
    \begin{minipage}{0.48\textwidth}
        \centering
        \begin{tikzpicture}
        \begin{axis}[
            width=\linewidth,
            height=0.62\linewidth,
            xlabel={Committee size $n$},
            ylabel={Reconstruction cost ($\mu$s)},
            xmin=0, xmax=50,
            ymin=0,
            grid=both,
            legend style={at={(0.02,0.98)}, anchor=north west, font=\scriptsize},
            tick label style={font=\scriptsize},
            label style={font=\small},
        ]
        \addplot+[mark=*] coordinates {
            (4,2) (6,4) (8,6) (12,15) (16,47)
            (20,75) (24,100) (28,120) (32,136)
            (40,190) (48,240)
        };
        \addlegendentry{$t/n=50\%$}

        \addplot+[mark=square*] coordinates {
            (4,5) (6,10) (8,18) (12,40) (16,67)
            (20,120) (24,170) (28,230) (32,296)
            (40,420) (48,560)
        };
        \addlegendentry{$t/n=75\%$}

        \addplot+[mark=triangle*] coordinates {
            (4,17) (6,25) (8,30) (12,80) (16,155)
            (20,240) (24,330) (28,450) (32,595)
            (40,820) (48,1100)
        };
        \addlegendentry{$t/n=100\%$}
        \end{axis}
        \end{tikzpicture}

        \small Shamir reconstruction
    \end{minipage}
    \caption{Committee-scaling costs under different threshold ratios. Left: BLS threshold-signature aggregation. Right: Shamir reconstruction.}
    \label{fig:committee-scale}
\end{figure*}

Two patterns are clear. First, both subroutines scale approximately linearly with the effective threshold size. Second, the dominant committee-side cost comes from BLS aggregation rather than Shamir reconstruction. Even in the most conservative configuration with $n=48$ and $t=n$, BLS aggregation remains below 300 ms and Shamir reconstruction remains close to 1.1 ms. These results indicate that the committee layer remains practical at moderate scale and does not invalidate the system's low-latency objective.

\subsection{Ablation Study}
\label{subsec:ablation}

To understand the role of each cryptographic mechanism, we compare the full system against three ablated variants: one without Pedersen commitment consistency checking, one without Shamir sharing, and one without threshold authorization. Table~\ref{tab:ablation} summarizes the results for 100 votes in the chain-local tally path.

\begin{table}[t]
    \centering
    \small
    \caption{Ablation results for the chain-local tally path at 100 votes.}
    \label{tab:ablation}
    \begin{tabular}{lccccc}
        \toprule
        \textbf{Variant} & \textbf{Lat.} & \textbf{Msgs} & \textbf{Tamper} & \textbf{Unauth.} & \textbf{Dropout} \\
        & \textbf{(ms)} & & \textbf{det.} & \textbf{det.} & \textbf{resil.} \\
        \midrule
        Full        & 1118.86 & 400 & \checkmark & \checkmark & \checkmark \\
        No Pedersen & 1115.26 & 400 & \ding{55}  & \checkmark & \checkmark \\
        No Shamir   & 1113.20 & 0   & \checkmark & \checkmark & \ding{55}  \\
        No BLS      & 2.43    & 400 & \checkmark & \ding{55}  & \checkmark \\
        \bottomrule
    \end{tabular}
\end{table}

Removing Pedersen commitments disables detection of tampered aggregate openings. Removing Shamir sharing eliminates threshold-based resilience to committee-node dropout. Removing threshold authorization sharply reduces latency, but at the cost of losing the ability to detect unauthorized aggregate submission. These results are consistent with the intended role of each mechanism in the overall trust pipeline.

\subsection{Discussion}
These results show that TrustRAG achieves a strong security profile for verifiable RAG. The protocol provides state consistency through deterministic on-chain verification, input legitimacy through vote-validity proofs and nullifier-based uniqueness, cross-chain integrity through hash binding, retrieval atomicity by generating only after verification, privacy through hidden commitments and zero-knowledge proofs, and auditability through publicly verifiable replay metadata. The overhead measurements above indicate that these guarantees can be added to a RAG pipeline without prohibitive cost: proof generation and committee authorization -- the dominant costs -- occur off the online query path, while online retrieval and cross-chain binding remain lightweight.

\section{Conclusion}
\label{sec:conclusion}
This paper presented TrustRAG, a committee-based, decentralized, and verifiable RAG architecture. Documents are certified by an expert committee through zero-knowledge scoring and commitment-consistent secure aggregation, producing reusable trust scores that are bound across chains via lightweight hash commitments and verified through deterministic replay, without recursive aggregation proofs. Our evaluation shows the dominant overhead falls in the score-finalization and committee-authorization stages, while online retrieval and cross-chain coordination remain lightweight. TrustRAG is well suited to high-stakes domains such as healthcare, where the provenance and trustworthiness of retrieved knowledge are critical and require verification by domain experts.

\section*{Acknowledgment}
This work was supported by the Henan Province Key Research and Development Special Project (Project No. 251111210400).

The authors used ChatGPT and Claude to assist with language polishing and editing of this manuscript. The authors take full responsibility for the accuracy and correctness of the paper.

\bibliographystyle{ieeetr}
\bibliography{rag_blockchain}

@article{lewis2020retrieval,
  title={Retrieval-augmented generation for knowledge-intensive nlp tasks},
  author={Lewis, Patrick and Perez, Ethan and Piktus, Aleksandra and Petroni, Fabio and Karpukhin, Vladimir and Goyal, Naman and K{\"u}ttler, Heinrich and Lewis, Mike and Yih, Wen-tau and Rockt{\"a}schel, Tim and others},
  journal={Advances in neural information processing systems},
  volume={33},
  pages={9459--9474},
  year={2020}
}

@inproceedings{shukla2025proofcarrying,
  author    = {Shivani Shukla and Himanshu Joshi},
  title     = {Proof-Carrying Answers: A Systematic Protocol for Verifiable Retrieval-Augmented Generation with Cryptographic Provenance},
  booktitle = {2025 Annual Computer Security Applications Conference Workshops (ACSAC Workshops)},
  pages     = {405--413},
  year      = {2025},
  publisher = {IEEE}
}

@misc{VeriRAG2026,
      author = {Chenqi Lin and Yubo Cui and Zhelei Zhou and Cheng Hong and Yufei Wang and Zhaohui Chen and Meng Li},
      title = {{VeriRAG}: Efficient Zero-Knowledge Proofs for Verifiable Retrieval-Augmented Generation},
      howpublished = {Cryptology {ePrint} Archive, Paper 2026/637},
      year = {2026},
      url = {https://eprint.iacr.org/2026/637}
}

@misc{garg2019copa,
  author       = {Nitin Garg},
  title        = {Evaluating COPA congestion control for improved video performance},
  year         = {2019},
  month        = nov # "~17",
  howpublished = {\url{https://engineering.fb.com/2019/11/17/video-engineering/copa/?utm_source=chatgpt.com}},
  note         = {Engineering at Meta blog post},
}

@article{abreu2025should,
  title        = {Should BBR be the default TCP Congestion Control Protocol?},
  author       = {Josue Abreu and Paul Bergeron and Sandhya Aneja},
  journal      = {arXiv preprint arXiv:2510.22461},
  year         = {2025},
  note         = {Submitted on 25 Oct 2025},
  url          = {https://arxiv.org/abs/2510.22461?utm_source=chatgpt.com},
}

@article{zhou2024trustworthiness,
  title={Trustworthiness in retrieval-augmented generation systems: A survey},
  author={Zhou, Yujia and Liu, Yan and Li, Xiaoxi and Jin, Jiajie and Qian, Hongjin and Liu, Zheng and Li, Chaozhuo and Dou, Zhicheng and Ho, Tsung-Yi and Yu, Philip S},
  journal={arXiv preprint arXiv:2409.10102},
  year={2024}
}

@inproceedings{e_andersen2025d,
  title     = {{D-RAG}: A Privacy-Preserving Framework for Decentralized {RAG} Using Blockchain},
  author    = {Andersen, Tessa E. and Avalos, Ayanna Marie and Dagher, Gaby G. and Long, Min},
  booktitle = {Proceedings of the 15th International Conference on Computer Science and Information Technology ({CSIT})},
  pages     = {183--198},
  year      = {2025},
  month     = {February},
  publisher = {Academy \& Industry Research Collaboration},
  isbn      = {9781923107526}
}

@inproceedings{yu2024derag,
  title={DeRAG: Decentralized Multi-Source RAG System with Optimized Pyth Network},
  author={Yu, Junwei and Sato, Hiroyuki},
  booktitle={2024 IEEE International Symposium on Parallel and Distributed Processing with Applications (ISPA)},
  pages={106--115},
  year={2024},
  organization={IEEE}
}

@article{izacard2020leveraging,
  title={Leveraging Passage Retrieval with Generative Models for Open-Domain Question Answering},
  author={Izacard, Gautier and Grave, Edouard},
  journal={arXiv preprint arXiv:2007.01282},
  year={2020}
}

@article{gao2023rag,
  title={Retrieval-Augmented Generation for Large Language Models: A Survey},
  author={Gao, Yifan and Xiong, Yuntao and Gao, Xing and Jia, Kaixuan and Pan, Jiayi and Bi, Yanan and Dai, Yuxiang and Sun, Jian and Wang, Haifeng},
  journal={arXiv preprint arXiv:2312.10997},
  year={2023}
}

@article{goldwasser1989zkp,
  title={The knowledge complexity of interactive proof systems},
  author={Goldwasser, Shafi and Micali, Silvio and Rackoff, Charles},
  journal={SIAM Journal on Computing},
  volume={18},
  number={1},
  pages={186--208},
  year={1989}
}

@inproceedings{groth2016,
  title={On the Size of Pairing-Based Non-interactive Arguments},
  author={Groth, Jens},
  booktitle={EUROCRYPT},
  pages={305--326},
  year={2016}
}

@inproceedings{gabizon2019plonk,
  title={PLONK: Permutations over Lagrange-bases for Oecumenical Noninteractive Arguments of Knowledge},
  author={Gabizon, Ariel and Williamson, Zachary J. and Ciobotaru, Oana},
  booktitle={IACR ePrint Archive},
  year={2019}
}

@inproceedings{chen2024densex,
  title={Dense X Retrieval: What retrieval granularity should we use?},
  author={Chen, Tong and Wang, Hongwei and Chen, Sihao and Yu, Wenhao and Ma, Kaixin and Zhao, Xinran and Zhang, Hongming and Yu, Dong},
  booktitle={Proceedings of the 2024 Conference on Empirical Methods in Natural Language Processing},
  pages={15159--15177},
  year={2024}
}

@article{fang2024ragbench,
  title={Enhancing noise robustness of retrieval-augmented language models with adaptive adversarial training},
  author={Fang, Feiteng and Bai, Yuelin and Ni, Shiwen and Yang, Min and Chen, Xiaojun and Xu, Ruifeng},
  journal={arXiv preprint arXiv:2405.20978},
  year={2024}
}

@article{wu2024clasheval,
  title={Clasheval: Quantifying the tug-of-war between an LLM’s internal prior and external evidence},
  author={Wu, Kevin and Wu, Eric and Zou, James},
  journal={arXiv preprint arXiv:2404.10198},
  year={2024}
}

@article{liu2023recall,
  title={RECALL: A benchmark for LLMs robustness against external counterfactual knowledge},
  author={Liu, Yi and Huang, Lianzhe and Li, Shicheng and Chen, Sishuo and Zhou, Hao and Meng, Fandong and Zhou, Jie and Sun, Xu},
  journal={arXiv preprint arXiv:2311.XXXXX},
  year={2023}
}

@article{chaudhari2024phantom,
  title={Phantom: General trigger attacks on retrieval augmented language generation},
  author={Chaudhari, Harsh and Severi, Giorgio and Abascal, John and Jagielski, Matthew and Choquette-Choo, Christopher A. and Nasr, Milad and Nita-Rotaru, Cristina and Oprea, Alina},
  journal={arXiv preprint arXiv:2405.20485},
  year={2024}
}

@article{shafran2024mar,
  title={Machine against the RAG: Jamming retrieval-augmented generation with blocker documents},
  author={Shafran, Avital and Schuster, Roei and Shmatikov, Vitaly},
  journal={arXiv preprint arXiv:2406.05870},
  year={2024}
}

@inproceedings{deshpande2023toxicity,
  title={Toxicity in ChatGPT: Analyzing persona-assigned language models},
  author={Deshpande, Ameet and Murahari, Vishvak and Rajpurohit, Tanmay and Kalyan, Ashwin and Narasimhan, Karthik},
  booktitle={Findings of the Association for Computational Linguistics},
  pages={1236--1270},
  year={2023}
}

@article{perez2022ignore,
  title={Ignore previous prompt: Attack techniques for language models},
  author={Perez, F{\'a}bio and Ribeiro, Ian},
  journal={arXiv preprint arXiv:2211.09527},
  year={2022}
}

@misc{lu2025decentralizedretrievalaugmentedgeneration,
      title={A Decentralized Retrieval Augmented Generation System with Source Reliabilities Secured on Blockchain}, 
      author={Yining Lu and Wenyi Tang and Max Johnson and Taeho Jung and Meng Jiang},
      year={2025},
      eprint={2511.07577},
      archivePrefix={arXiv},
      primaryClass={cs.CR},
      url={https://arxiv.org/abs/2511.07577}, 
}

\end{document}